\documentclass{article}
\usepackage{graphicx} 
\usepackage{styling}
\usepackage[T1]{fontenc}
\usepackage{natbib}

\usepackage[colorlinks]{hyperref}
\usepackage[ruled,linesnumbered]{algorithm2e}
\usepackage{booktabs}

\author{%
  Praneeth Kacham\thanks{Work done while the author was a student at Carnegie Mellon University.} \\
  Google Research\\
  \texttt{pkacham@google.com} \\
  \and
  David P. Woodruff\\
  Carnegie Mellon University\\
  \texttt{dwoodruf@cs.cmu.edu}
}

\DeclareMathOperator*{\argmax}{arg\,max}

\title{Approximating the Top Eigenvector in Random Order Streams}

\date{}
\begin{document}
\maketitle
\begin{abstract}
    When rows of an $n \times d$ matrix $A$ are given in a stream, we study algorithms for approximating the top eigenvector of the matrix $\T{A}A$ (equivalently, the top right singular vector of $A$). We consider worst case inputs $A$ but assume that the rows are presented to the streaming algorithm in a uniformly random order. We show that when the gap parameter $R = \sigma_1(A)^2/\sigma_2(A)^2 = \Omega(1)$, then there is a randomized algorithm that uses $O(h \cdot d \cdot \polylog(d))$ bits of space and outputs a unit vector $v$ that has a correlation $1 - O(1/\sqrt{R})$ with the top eigenvector $v_1$. Here $h$ denotes the number of \emph{heavy rows} in the matrix, defined as the rows with Euclidean norm at least $\frnorm{A}/\sqrt{d \cdot \polylog(d)}$. We also provide a lower bound showing that any algorithm using $O(hd/R)$ bits of space can obtain at most $1 - \Omega(1/R^2)$ correlation with the top eigenvector. Thus, parameterizing the space complexity in terms of the number of heavy rows is necessary for high accuracy solutions.

    Our results improve upon the $R = \Omega(\log n \cdot \log d)$ requirement  in a recent work of Price and Xun (FOCS 2024). We note that the  algorithm of Price and Xun works for arbitrary order streams whereas our algorithm requires a stronger assumption that the rows are presented in a uniformly random order. We additionally show that the gap requirements in their analysis can be brought down to $R = \Omega(\log^2 d)$ for arbitrary order streams and $R = \Omega(\log d)$ for random order streams. The requirement of $R = \Omega(\log d)$ for random order streams is nearly tight for their analysis as we obtain a simple instance with $R = \Omega(\log d/\log\log d)$ for which their algorithm, with any fixed learning rate, cannot output a vector approximating the top eigenvector $v_1$.
\end{abstract}
\section{Introduction}
We consider the problem of approximating the top eigenvector in the streaming setting. In this problem, we are given vectors $a_1, \ldots, a_n \in \R^d$ one at a time in a stream. Let $A$ be an $n \times d$ matrix with rows $a_1, \ldots, a_n$. The task is to approximate the top eigenvector of the matrix $\T{A}A$. Throughout the paper, we use $v_1 \in \R^d$ to denote the top eigenvector of $\T{A}A$. We focus on obtaining streaming algorithms that use a small amount of space and can output a unit vector $\hat{v}$ such that $\la \hat{v}, v_1\ra^2 \ge 1 - f(R)$, where $f(R)$ is a decreasing function in the gap $R = \lambda_1(\T{A}A)/\lambda_2(\T{A}A)$. Here $\lambda_1(\cdot), \lambda_2(\cdot)$ denote the two largest eigenvalues. As the gap $R$ becomes larger, the eigenvector approximation problem becomes \emph{easier} and we therefore want more accurate approximations to the eigenvector $v_1$.

If one is allowed to use $\tilde{O}(d^2)$\footnote{The notation $\tilde{O}(f(n))$ is used to denote the set of functions in $O(f(n) \cdot \polylog(n))$.} bits of space, we can maintain the matrix $\T{A}A = \sum_i a_i\T{a_i}$ as we see the rows $a_i$ in the stream, and at the end of processing the stream, we can compute the exact top eigenvector $v_1$. When the dimension $d$ is large, the requirement of $\Omega(d^2)$ bits of memory can be impractical (see e.g., applications that require a large value of $d$ in \cite{mitliagkas2013memory}.) Hence, an interesting question is to study non-trivial streaming algorithms that use less memory. In this work, we focus on obtaining algorithms that use $\tilde{O}(d)$ bits of space.

In the offline setting (where the entire matrix $A$ is available to us), fast iterative algorithms such as \citet{gu2015subspace, musco2015randomized, musco2018stability} can be used to quickly obtain accurate approximations to the top eigenvector when the gap $R = \Omega(1)$. In a single pass streaming setting, we cannot run these algorithms as these iterative algorithms need to \emph{see} the entire matrix multiple times.

There have been two major lines of work studying the problem of eigenvector approximation and the related Principal Component Analysis (PCA) problem in the streaming setting with near-linear in $d$ memory. In the first line of work, each row encountered in the stream is assumed to be sampled independently from an unknown distribution with mean $0$ and covariance $\Sigma$ and the task is to approximate the top eigenvector of $\Sigma$ using the samples. In this line of work, the sample complexity required for algorithms using $O(d \cdot \polylog(d))$ bits of space to output an approximation to ${v}_1$, is the main question. The algorithms are usually a variant of Oja's algorithm \citep{oja1982simplified, jain2016streaming, allen2017first, huang2021streaming, kumar2024streaming} or the block power method \citep{hardt2014noisy, balcan2016improved}. We note that \citet{kumar2024streaming} relax the i.i.d. assumption and analyze the sample complexity of Oja's algorithm for estimating the top eigenvector in the Markovian data setting.

The other line of work studies algorithms for arbitrary streams appearing in an arbitrary order. In this setting, we want algorithms to work for \emph{any} input stream given in \emph{any} order. A problem closely related to the eigenvector estimation problem is the Frobenius-norm Low Rank Approximation \citep{clarkson2017low, boutsidis2016optimal, upadhyay2016fast, ghashami2016frequent}. The deterministic Frequent Directions sketch of \citet{ghashami2016frequent} can, using $\tilde{O}(d/\varepsilon)$ bits of space, output a unit vector $u$ such that
\begin{align*}
    \frnorm{A(I-u\T{u})}^2 \le (1 + \varepsilon)\frnorm{A(I-v_1\T{v_1})}^2.
\end{align*}
Although the vector $u$ is a $1+\varepsilon$ approximate solution to the Frobenius norm Low Rank Approximation problem, it is possible that the vector $u$ may be (nearly) orthogonal to the top eigenvector $v_1$. Hence the Frequent Directions sketch does not guarantee top eigenvector approximation. Recently, \citet{price2023spectral} study the eigenvector approximation problem in arbitrary streams and obtain results in terms of the gap $R$ of the instance. Price and Xun prove that when $R = \Omega(\log n \cdot \log d)$, a variant of Oja's algorithm outputs a unit vector $\hat{v}$ such that 
\begin{align*}
    \la \hat{v}, v_1\ra^2 \ge 1 - \frac{C\log d}{R} - \frac{1}{\poly(d)}
\end{align*}
where $C$ is a large enough universal constant. On the lower bound side, Price and Xun showed that any algorithm that outputs a vector $\hat{v}$ satisfying
\begin{align*}
    \la \hat{v}, v_1\ra^2 \ge 1 - \frac{1}{CR^2},
\end{align*}
must use $\Omega(d^2/R^3)$ bits of space while processing the stream. This lower bound shows that in the important case of $R = O(1)$, the \emph{correlation}\footnote{We say that the value $\la u, v\ra^2$ denotes the correlation between unit vectors $u$ and $v$.} that can be obtained by an algorithm using $\tilde{O}(d)$ bits of space is at most a constant less than $1$. Thus, the current best algorithms for arbitrary streams work only when $R = \Omega(\log n \cdot \log d)$ and for the important case of $R = O(1)$, there are no existing algorithms requiring significantly fewer than $d^2$ bits of memory. They also give a lower bound on the size of \emph{mergeable} summaries for approximating the top eigenvector.

We identify an instance with $R = \Theta(\log d/\log\log d)$ where the algorithm of Price and Xun fails to produce a vector with even a constant correlation with the vector $v_1$. This shows that their algorithm or other variants of Oja's algorithm may fail to extend to the case when $R = O(1)$. We further show that the algorithm of Price and Xun fails to produce such a vector even when the rows in our hard instance are ordered uniformly at random, showing that even randomly ordered  streams can be hard to solve for variants of Oja's algorithm.

In this work, we focus on algorithms that work on worst case inputs $A$ while assuming that the rows of $A$ are \emph{uniformly randomly ordered}. This model is mid-way between the i.i.d. setting and the arbitrary order stream setting in terms of the generality of streams that can be modeled. We note that a number of works \citep{munro1980selection, guha2005streaming, chakrabarti2008robust, guha2009stream, assadi2023noisy} have previously considered streaming algorithms and lower bounds for worst case inputs with random order streams, as it is a natural model often arising in practical settings.  See \cite{GuptaSingla2021} for a gentle introduction to the random-order model. Our algorithms are parameterized in terms of the number of \textbf{heavy} rows in the stream. We define a row $a_i$ to be \emph{heavy} if $\opnorm{a_i} \ge \frnorm{A}/\sqrt{d \cdot \polylog(d)}$. Note that in any stream of rows, by definition, there are at most $d \cdot \polylog(d)$ heavy rows. We state our theorem informally below:
\begin{theorem}
    Let $a_1, \ldots, a_n \in \R^d$ be a randomly ordered stream and let $A$ denote the $n \times d$ matrix with rows given by $a_1, \ldots, a_n$. If $R = \lambda_1(\T{A}A)/\lambda_2(\T{A}A) > C$ for a large enough constant $C$ and the number of heavy rows in the stream is at most $h$, then there is a streaming algorithm using $O(h \cdot d \cdot \polylog(d))$ bits of space and outputting a unit vector $\hat{v}$ satisfying
    \begin{align*}
        \la\hat{v}, v_1\ra^2 \ge 1 - O(1/\sqrt{R})
    \end{align*}
    with a probability $\ge 4/5$.
\end{theorem}
Our algorithm is a variant of the block power method. Along the way, we also improve the gap requirements in the results of \citet{price2023spectral}. We show that by subsampling a stream of rows, the algorithm of Price and Xun can be made to work even when the gap $R$ is $\Omega(\log^2 d)$ in arbitrary order streams, improving on the $\Omega(\log n \cdot \log d)$ requirement in their analysis. We also show that in random order streams, a gap of $\Omega(\log d)$ is sufficient for their algorithm, though our algorithm improves on this and works for even a constant gap. 

Similar to the lower bound of Price and Xun, we show that any algorithm for random order streams must use $\Omega(h \cdot d / R)$ bits of space to output a vector $\hat{v}$ satisfying $\la \hat{v}, v_1\ra^2 \ge 1 - 1/CR^2$ where $C$ is a constant. We summarize the theorem below.
\begin{theorem}
    Consider an arbitrary random order stream $a_1, \ldots, a_n$ with the gap parameter $\frac{\sigma_1(A)^2}{\sigma_2(A)^2} = R$. Let $h$ be the number of \emph{heavy} rows in the stream. Any streaming algorithm that outputs a unit vector $\hat{v}$ such that 
    \begin{align*}
        \la \hat{v}, v_1\ra^2 \ge 1 - 1/CR^2
    \end{align*}
    for a large enough constant $C$, with a probability $\ge 1 - (1/2)^{R+1}$ over the ordering of the stream and its internal randomness, must use $\Omega(h \cdot d / R)$ bits of space.
    \label{thm:informal-lowerbound}
\end{theorem}
\paragraph{Techniques.} The randomized power method \citep{gu2015subspace} algorithm to approximate the top eigenvector samples a random Gaussian vector $\bg$ and iteratively computes the vector $v = (\T{A}A)^t \bg$\footnote{Note that $\T{A}A \cdot v = \sum_i \la a_i, v\ra a_i$.} for $t = \Theta(\log d)$ iterations and shows that when the gap $R$ is large, $v/\opnorm{v}$ is a good approximation for $v_1$. Thus, the algorithm needs to \emph{see} the quadratic form $\T{A}A$ multiple times and hence, it cannot be implemented in the single-pass streaming setting of this paper.

Assume that the stream is randomly ordered and that there are no heavy rows. Our key observation is that if the stream is long enough, then we can see $t$ approximations $\T{\bB_j}\bB_j$\footnote{We use bold symbols to denote random variables.} of the quadratic form $\T{A}A$. Here the matrices $\bB_1, \ldots, \bB_t$ are formed by sampling and rescaling the rows of the matrix $A$ and importantly, the rows of $\bB_1, \ldots, \bB_t$ do not overlap in the stream, that is, they appear one after the other. Thus we can compute $v' = (\T{\bB_t}\bB_t) \cdots (\T{\bB_1}\bB_1) \cdot \bg$ for the starting vector $\bg$ in a single pass over the stream. We prove that such matrices $\bB_j$ exist using the row norm sampling result of \cite{magdon2010row}. Now, the main issue is to show that $v'/\opnorm{v'}$ is a good approximation to the top eigenvector $v_1$. We crucially use a singular value inequality of \cite{wang1997some} to prove that $\opnorm{\T{\bB_j}\bB_j - \T{A}A} \le \varepsilon\opnorm{A}^2$ for all $j$ suffices for $v'/\opnorm{v'}$ to be a good approximation to $v_1$.

The above analysis assumes that there are no heavy rows. Indeed, suppose that a matrix $A$ has a row $a$ with a large Euclidean norm which is orthogonal to all the other rows. Also assume that the top eigenvector of the matrix $A$ is in this direction. Since, the matrices $\bB_1, \ldots, \bB_t$ are non-overlapping substreams of the matrix $A$, at most one of the matrices $\bB_j$ can have the row $a$ and hence the vector $v'/\opnorm{v'}$ will not be a good approximation to $a/\opnorm{a}$, the top eigenvector. Thus, we need to handle the heavy rows separately. We show that, by storing all the rows with a Euclidean norm larger than $\frnorm{A}/\sqrt{d \cdot \polylog(d)}$ and running the above described algorithm on the remaining set of rows, we can obtain a good approximation to the top eigenvector.

Our lower bound (Theorem~\ref{thm:informal-lowerbound}) shows that any single-pass streaming algorithm must use space proportional to the number of heavy rows, and therefore our procedure that handles the heavy rows separately gives near-optimal bounds. 

Finally, the row norm sampling technique of \cite{magdon2010row} serves as a general technique to reduce the number of rows in the stream while (approximately) preserving the top eigenvector. We use this observation to improve the $R = \Omega(\log n \cdot \log d)$ for arbitrary streams in \cite{price2023spectral} to $R = \Omega(\log^2 d)$. We then show that assuming a uniformly random order, the analysis of \cite{price2023spectral} can be improved to show that $R = \Omega(\log d)$ suffices. Thus, for random order streams, techniques before our work can be used to approximate the top eigenvector when the gap $R = \Omega(\log d)$. Our work improves upon this to give an algorithm that works for streams with $R = \Omega(1)$.

\paragraph{Implications to practice.} Often, in practical situations, we can assume that the rows being streamed are sampled independently from a nice-enough distribution, in which case Oja's algorithm, as discussed, can approximate the top eigenvector accurately given enough samples. However, \emph{independence} and assumptions on the covariance matrix can be very strong assumptions in some cases and in such cases, our algorithm only requires that the order of the rows in the stream be uniformly random, in which case we output an approximation with provable guarantees.

\paragraph{Organization.} We first introduce the row-norm sampling procedure to obtain approximate quadratic forms. The proof is a slight modification of that of \cite{magdon2010row}. The only difference is that we instead consider a version that samples each row in the input independently with some appropriate probability and keeps the rows that are sampled after scaling appropriately. We then introduce and analyze our block power iteration algorithm when all rows have roughly the same Euclidean norm, and then extend it to the general case, which is our main result. Finally, we provide a lower bound showing that $\Omega(td/R)$ bits of space is necessary to obtain constant correlation with the top eigenvector. 

\section{Power Method with Approximate Quadratic Forms}
In this section, we present and analyze our algorithm for approximating the top eigenvector of $\T{A}A$ when the rows of $A$ are presented to the algorithm in a uniformly random order.

We first show a row sampling technique that reduces the number of rows in the stream. The row-norm sampling technique for approximating the quadratic form $\T{A}A$ with spectral norm guarantees was given by  \cite{magdon2010row}. The technique works irrespective of the order of the rows.
\subsection{Sampling for Row Reduction}\label{sec:sampling-for-row-reduction}
\begin{theorem}
	Let $A$ be an arbitrary $n \times d$ matrix. Given $p \in [0, 1]^n$, let $\bQ$ be an $n \times n$ diagonal matrix such that for each $i \in [n]$, we independently set $\bQ_{ii} = 1/\sqrt{p_i}$ with probability $p_i$ and $0$ otherwise. If for all $i$,
	\begin{align*}
		p_i \ge \min\left(1, C\frac{\opnorm{a_i}^2}{\varepsilon^2\opnorm{A}^2}\log d\right),
	\end{align*}
	then with probability $1 - 1/\poly(d)$, 
	$
		\opnorm{\T{A}A - \T{A}\T{\bQ}\bQ A} \le \varepsilon\opnorm{A}^2.
	$
	With probability at least $1 - 1/\poly(d)$, the matrix $\bQ$ has at most $O(\varepsilon^{-2}\rho \log d)$ non-zero entries, where $\rho = \frnorm{A}^2/\opnorm{A}^2$ denotes the stable rank of matrix $A$.
 \label{thm:row-norm-sampling}
\end{theorem}
\begin{proof}
Let $\bX_i$ denote an indicator random variable which denotes if $\bQ_{ii}$ is nonzero. Note $\E[\bX_i] = p_i$ and $\bX_1, \ldots, \bX_n$ are independent. Define a $d \times d$ random matrix $\bY_i = (\bX_i/p_i - 1)a_i\T{a_i}$, where $a_i$ denotes the $i$-th row of $A$. We note that
\begin{align*}
	\T{A}A - \T{A}\T{\bQ}\bQ A = \sum_{i=1}^n (\bX_i/p_i - 1)a_i\T{a_i} = \sum_{i=1}^n \bY_i.
\end{align*}
We use the Matrix Bernstein inequality \citep{tropp2015introduction} to bound $\opnorm{\sum_i \bY_i}$. We first uniformly upper bound $\opnorm{\bY_i}$. If $p_i = 1$, by definition $\opnorm{\bY_i} = 0$ with probability $1$. Let $p_i \ne 0$. Then, $\opnorm{(\bX_i/p_i - 1)a_i\T{a_i}} \le \opnorm{a_i\T{a_i}}/p_i \le {\varepsilon^2\opnorm{A}^2}/{C\log d}$
with probability $1$. 

We now bound $\opnorm{\sum_i \E[\bY_i^2]}$. 
\begin{align*}
	\sum_i \E[\bY_i^2] &= \sum_i \E[(1/p_i - 1)^2] \opnorm{a_i}^2 a_i\T{a_i}\\
	&=\sum_{i : p_i > 0} (1/p_i - 1)\opnorm{a_i}^2a_i\T{a_i}\\
	&\preceq \sum_{i: p_i > 0} \frac{\varepsilon^2\opnorm{A}^2}{C\opnorm{a_i}^2\log d}\opnorm{a_i}^2a_i\T{a_i}\\
	&\preceq \frac{\varepsilon^2\opnorm{A}^2}{C\log d}\T{A}A
\end{align*}
which implies $\opnorm{\sum_i \E[\bY_i^2]} \le \varepsilon^2\opnorm{A}^4/(C\log d)$. Now, we obtain
\begin{align*}
	\Pr[\opnorm{\sum_i \bY_i} \ge \varepsilon\opnorm{A}^2] &\le 2d \cdot \exp\left(-\frac{\varepsilon^2\opnorm{A}^4/2}{\varepsilon^2\opnorm{A}^4/(C\log d) + \varepsilon^3\opnorm{A}^4/(3C\log d)}\right)\\
	&\le 2d \cdot \exp\left(-\frac{C\log d}{2(1+\varepsilon/3)}\right).
\end{align*}
If $C \ge 6(1+\varepsilon/3)$, then $\Pr[\opnorm{\sum_i \bY_i} \ge \varepsilon\opnorm{A}^2] \le 1 - 2/d^2$ which implies that with probability $\ge 1 - 2/d^2$, $\opnorm{\T{A}A - \T{A}\T{\bQ}\bQ A} \le \varepsilon\opnorm{A}^2$.

Now, the number of non-zero entries in the matrix $\bQ$ is equal to $\sum_i \bX_i$. We note $\E[\sum_i \bX_i] \le C\varepsilon^{-2}\rho\cdot \log d$. By a Chernoff bound, we obtain that $\sum_i \bX_i = O(\varepsilon^{-2}\rho\cdot \log d)$ with probability $\ge 1 - 1/\poly(d)$.
\end{proof}

Note that given the value of $\opnorm{A}$, the sampling procedure in this theorem can be performed in a stream. Additionally, as the original stream is uniformly randomly ordered, the sub-sampled stream is also uniformly randomly ordered assuming that the sampling is independent of the order of the rows. 

Given that all of the non-zero entries of the matrix have absolute value at least $1/\poly(nd)$ and at most $\poly(nd)$, we have that $\opnorm{A}^2$ lies in the interval $[1/\poly(nd), \poly(nd)]$. Thus, we can guess the value of $\opnorm{A}^2$ as $2^{i}/\poly(nd)$ for $i = 0, \ldots, O(\log(nd))$ and one of these values must be a $2$-approximation to $\opnorm{A}^2$, and thus sub-sampling the rows using that guess satisfies the conditions in the above theorem. We can run the streaming algorithms on all the streams simultaneously to obtain $O(\log nd)$ vectors $u_1, \ldots, u_{O(\log nd)}$ as the candidates for being an approximation to the top eigenvector. From Theorem~\ref{thm:row-norm-sampling}, the candidate vector $u_j$ computed on the stream obtained by sampling the rows with the correct probabilities is a good approximation to the top eigenvector, and therefore $\opnorm{A \cdot u_j}$ is large for that value of $j$. Thus, the vector $u_j$ with the largest value $\opnorm{A \cdot u_j}$ is a good approximation to the top eigenvector $v_1$. If $\bG$ is a Gaussian matrix with $O(\varepsilon^{-2}\log d)$ rows, then for all $u_j$, we can approximate $\opnorm{A \cdot u_j}$ up to a $1 \pm \varepsilon$ factor using $\opnorm{\bG \cdot A \cdot u_j}$ by the Johnson-Lindenstrauss lemma. Additionally, the matrix $\bG \cdot A$ can be maintained in the stream using $O(\varepsilon^{-2} \cdot d \log d)$ bits (when we see a row $a_i$, we sample an independent Gaussian vector $\bg_i$ and add $\bg_i \T{a_i}$ to an accumulator to maintain $\bG \cdot A$). Thus, at the end of processing the stream, we can compute a vector $u_j$ that has a large value $\opnorm{A \cdot u_j}$,  and hence is a good approximation for $v_1$.

If we can process each created stream using $s$ bits of space, then the overall space requirement is $O(s \cdot \log(nd) + d \cdot \polylog(d))$ bits, using $O(s)$ bits for each guess for the value of $\opnorm{A}^2$ and $O(d \cdot \polylog(d))$ bits for storing a Gaussian sketch of the matrix with $\varepsilon = 1/\polylog(d)$.
\subsection{Random-Order Streams with bounds on Norms}\label{sec:random-order-streams-bounded-norms}
\begin{algorithm}
\caption{Approximate Eigenvector for Streams with no Large Norms}
\label{alg:bounded-norms}
\KwIn{An $n \times d$ matrix $A$ with $n = \Omega(\eta \cdot \rho(A) \cdot \log^2 d/\varepsilon^2)$, $\max_i \opnorm{a_i}^2/\min_i \opnorm{a_i}^2 \le \eta$}
\KwOut{A vector $\bz$}
\DontPrintSemicolon
$t \gets \ceil{C_1\log d}$\; 
Compute $\bG \cdot A$ in the stream where $\bG$ is a Gaussian matrix with $O(\varepsilon^{-2}\log d)$ rows\;
\For{$\rho = 1, 2, 4, \ldots, d$ simultaneously}{
$p \gets C_2\eta\rho\log d/n\varepsilon^2$\tcp*{$p \le 1/(5t)$ for $\rho \le 2 \cdot \rho(A)$}
$\bz_{\rho} \sim N(0, 1)^d$\;
    \For{$j=1, \ldots, t$}{
        $\by_j \gets \text{Bin}(n, p)$\;
        \If{$\by_j > 2np$}{
            \Return{$\perp$}
        }
        \tcp{The matrix $A_{j \cdot (2np) : j \cdot (2np) + \by_j}$ corresponds to $\bB_j$ in the analysis.}
        $\textit{acc} \gets 0$\;
        \For{$i = (j-1) \cdot (2np) + 1, \ldots, (j-1) \cdot (2np) + \by_j$}{
            $\textit{acc} \gets \textit{acc} + \la a_i, \bz_{\rho}\ra \cdot a_i$\;
        }
        \tcp{Here $\textnormal{\textit{acc}} = \T{\bB_j}{\bB_j}\bz_{\rho}$}
        $\bz_{\rho} \gets \textit{acc}$\;
        $\bz_{\rho} \gets \bz_{\rho} / \opnorm{{\bz}_{\rho}}$\;
    }
}
\Return{$\argmax_{\bz \in \set{\bz_1, \bz_2, \bz_4, \ldots, \bz_d}}\opnorm{(\bG \cdot A)\bz}$}
\end{algorithm}

We now present the analysis of the block power method for random order streams assuming that the Euclidean norms of all the rows in $A$ are close to each other. We later remove this assumption. Suppose there exists a parameter $\eta$ such that $(\max_i \opnorm{a_i}^2)/(\min_i \opnorm{a_i}^2) \le \eta$. If $\eta$ is close to $1$ then all the rows in the stream have roughly the same norm. 

Let $p = C\eta \rho \log (d)/\varepsilon^2n$. We can see that for any row $a_i$ in the stream,
\begin{align*}
	C\frac{\opnorm{a_i}^2}{\varepsilon^2\opnorm{A}^2}\log d \le C\frac{\eta \frnorm{A}^2/n}{\varepsilon^2\opnorm{A}^2}\log d \le \frac{C\eta \rho\log d}{n\varepsilon^2} = p.
\end{align*}
Thus, $p$ is greater than the probability with which we need to sample each row in the row-norm sampling  result in Theorem~\ref{thm:row-norm-sampling}. Now if we perform such a sampling of the rows of $A$, we sample $\text{Bin}(n, p)$\footnote{$\text{Bin}(n, p)$ denotes the binomial distribution with parameters $n$ and $p$.} number of rows, which is tightly concentrated around $np = \varepsilon^{-2}C\eta \rho \log d$. Thus, if we first sample $\by \sim \text{Bin}(n, p)$ and then consider the first $\by$ number of rows in the random order stream, then we will have sampled from a distribution satisfying the requirements in Theorem~\ref{thm:row-norm-sampling} and can therefore obtain a matrix $\bB$ such that
\begin{align*}
	\opnorm{\T{\bB}\bB - \T{A}A} \le \varepsilon\opnorm{A}^2.
\end{align*}
Thus, assuming that the rows appear in a uniformly random order lets us show that the first $\by$ rows of the stream can be used to compute an approximation to the quadratic form $\T{A}A$. We will now show that we can obtain $O(\log d)$ such quadratic forms in the stream given that the stream is long enough.

Assume that the number of rows in the stream $n = \Omega(\eta \rho \log ^2 d/\varepsilon^2)$. We partition the stream into $t = \Theta(\log d)$ groups as follows: the first $2n p$ rows are placed in the group $1$, the second $2n p$ rows are placed in the group $2$, and so on. Note that since $n = \Omega(\eta \rho \log^2 d /\varepsilon^2)$, we can form $t$ such groups. Since the rows are uniformly randomly ordered, the joint distribution of the rows appearing in group $1$ is the same as that of the joint distribution of the rows appearing in group $2$ and so on. Let $\by_1, \ldots, \by_t \sim \text{Bin}(n, p)$ be drawn independently. With probability $\ge 1 - 1/\poly(d)$, we have $\by_i \le (3/2)np$ for all $i$. For $i = 1, \ldots, t$, let $\bB_i$ be the matrix formed by the first $\by_i$ rows in group $i$. Using a union bound, we have that with probability $\ge 1 - 1/\poly(d)$, for all $i = 1, \ldots, t$, 
\begin{align*}
    \opnorm{\T{A}A - \frac{1}{p}\T{\bB_i}{\bB_i}} \le \varepsilon\opnorm{A}^2.
\end{align*}
Conditioned on the above event, we will now show that running the power method on the blocks $\bB_1, \ldots, \bB_t$ lets us approximate the top singular vector of the matrix $A$.
\begin{assumption}
	We assume that $\sigma_1(A) / \sigma_2(A) \ge 2$.
\end{assumption}
\begin{lemma}
Let $\varepsilon > 1/\poly(d)$ be an accuracy parameter and $t = \Omega(\log d)$ be the number of iterations. Let $\varepsilon \le c/t^2$ for a small constant $c$. Suppose $B_1, \ldots, B_t$ all satisfy $\opnorm{\T{A}A - \T{B_j}B_j} \le \varepsilon\opnorm{A}^2$ for $\varepsilon < 1/5$. If $\bg$ is a random vector sampled from the Gaussian distribution, then the unit vector
	\begin{align*}
		\hat{v} \coloneqq \frac{(\T{B_t}B_t) \cdots (\T{B_1}B_1)\bg}{\opnorm{(\T{B_t}B_t) \cdots (\T{B_1}B_1)\bg}}
	\end{align*}
	satisfies
	\begin{align*}
		\langle\hat{v}, v_1\rangle^2 \ge \frac{1}{1 + C't\sqrt{\varepsilon}}	\end{align*}
	with probability $\ge 9/10$ for a large enough constant $C'$. Here $v_1$ denotes the top right singular vector of the matrix $A$.
	\label{lma:appx-quads-to-correlation}
\end{lemma}
To prove this lemma, our strategy is to show that the matrix product $M \coloneqq (\T{B_t}B_t) \cdots (\T{B_1}B_1)$ has a stable rank close to $1$ --- meaning it has one very large singular value and the rest of the singular values are small. We can then argue that the vector $\hat{v} = M\bg/\opnorm{M\bg}$ is in the direction of the top singular vector $M$. Using the fact that $\T{v_1}(\T{B_j}B_j)v_1 \ge (1-\varepsilon)\opnorm{A}^2$ for all $j$, we show that the top singular vector of $M$ must have a large correlation with $v_1$. Therefore, it follows that the vector $\hat{v}$ has a large correlation with $v_1$ as well. As part of the proof, we crucially use an inequality from \cite{wang1997some}. The full formal proof now follows.
\begin{proof}
	Define $M \coloneqq (\T{B_t}B_t) \cdots (\T{B_1}B_1)$. Our strategy is to show that if $v_1$ is the top singular vector of the matrix $A$, then $\opnorm{\T{v_1}M}$ is comparable to $\frnorm{M}$ given that $\sigma_1(A)/\sigma_2(A) \ge 2$. We can then prove the lemma using simple properties of the Gaussian vector $\bg$.

For an arbitrary $j$, let $(\T{B_j}B_j)v_1 = \alpha v_1 + \Delta$ where $\Delta \perp v_1$. We note that
$
		\T{v_1}(\T{B_j}B_j)v_1 = \alpha
$. We have $\alpha = \T{v_1}\T{B_j}B_jv_1 \ge (1 - \varepsilon)\sigma_1(A)^2$ using the fact that $\opnorm{\T{B_j}B_j - \T{A}A} \le \varepsilon\opnorm{A}^2$ and $\T{v_1}\T{A}Av_1 = \sigma_1(A)^2 = \opnorm{A}^2$. If we show that $\Delta$ is small, then the vector $(\T{B_j}B_j)v_1$ is oriented in a direction very close to that of $v_1$. Note that
\begin{align*}
	\opnorm{(\T{B_j}B_j)v_1} \le \opnorm{\T{B_j}B_j} \le (1+\varepsilon)\sigma_1(A)^2
\end{align*}
and $\opnorm{(\T{B_j}B_j)v_1}^2 = \alpha^2 + \opnorm{\Delta}^2$ which implies $\opnorm{\Delta}^2 \le ((1+\varepsilon)^2 - (1-\varepsilon)^2)\sigma_1(A)^4 = 4\varepsilon \cdot \sigma_1(A)^4$ and thus $\opnorm{\Delta} \le \sqrt{4\varepsilon}\sigma_1(A)^2$. Now,
\begin{align*}
	&\opnorm{\T{M}v_1}\\
 &= \opnorm{(\T{B_1}B_1) \cdots (\T{B_{t-1}}B_{t-1})(\langle \T{B_t}B_t v_1, v_1\rangle v_1 + \Delta_1)}\\
	&\ge \langle \T{B_t}B_tv_1, v_1\rangle\opnorm{(\T{B_1}B_1) \cdots (\T{B_{t-1}}B_{t-1})v_1} - \opnorm{(\T{B_1}B_1)\cdots(\T{B_{t-1}}B_{t-1})}\opnorm{\Delta_1}\\
	&\ge ((1-\varepsilon)\sigma_1(A)^2)\opnorm{(\T{B_1}B_1) \cdots (\T{B_{t-1}}B_{t-1})v_1} - (\sqrt{4\varepsilon}\sigma_1(A)^2)\opnorm{(\T{B_1}B_1) \cdots (\T{B_{t-1}}B_{t-1})}.
\end{align*}
Expanding similarly, we obtain
\begin{align*}
	\opnorm{\T{M}v_1} \ge (1 - \varepsilon)^t\sigma_1(A)^{2t} - t\sqrt{4\varepsilon}(1+\varepsilon)^{t-1}\sigma_1(A)^{2t}.
\end{align*}
Assuming $\varepsilon \le c/t$ for a small constant $c$, we note that $(1 - \varepsilon)^t \ge (1 - 2t\varepsilon)$ and $(1 + \varepsilon)^t \le (1 + 2t\varepsilon)$ which implies
\begin{align*}
	\opnorm{\T{M}v_1} = \opnorm{(\T{B_1}B_1) \cdots (\T{B_t}B_t)v_1} \ge (1 - 2t\varepsilon - 4t\sqrt{\varepsilon})\sigma_1(A)^{2t}.
\end{align*}
We shall now show a bound on $\frnorm{M} = \frnorm{(\T{B_t}B_t) \cdots (\T{B_1}B_1)}$ which lets us show that the unit vector $\hat{v}$ is highly correlated with $v_1$. To bound the quantity $\frnorm{M}$, we first note the following facts:
\begin{enumerate}
	\item $\opnorm{\T{B_j}B_j} \le (1 + \varepsilon)\sigma_1(A)^2$, and
	\item $\sigma_2(\T{B_j}B_j) \le \sigma_2(A)^2 + \varepsilon\sigma_1(A)^2 \le (1/4 + \varepsilon)\sigma_1(A)^2$ by our gap assumption.
\end{enumerate}
Now, we use the following theorem.
\begin{theorem}[{\cite[Theorem~3(ii)]{wang1997some}}]
	For any $r > 0$ and any matrices $A_1, \ldots, A_t$, 
	\begin{align*}
		\sum_i (\sigma_i(A_1 \cdots A_t))^r \le \sum_i \sigma_i(A_1)^r \cdots \sigma_i(A_t)^r.
	\end{align*}
\end{theorem}
Applying the above theorem with $r=2$, we obtain
\begin{align*}
	\frnorm{(\T{B_t}B_t) \cdots (\T{B_1}B_1)}^2 &\le (1+\varepsilon)^{2t}\sigma_1(A)^{4t} + (d-1)(1/4 + \varepsilon)^{t}\sigma_1(A)^{4t}\\
	&\le (1+4t\varepsilon)\sigma_1(A)^{4t} + \frac{d}{3^t}\sigma_1(A)^{4t}.
\end{align*}
When $t \ge 3\log (d/\varepsilon)$, we have 
$
	\frnorm{(\T{B_t}B_t) \cdots (\T{B_1}B_1)}^2 \le (1 + 4t\varepsilon + \varepsilon)\sigma_1(A)^{4t}.
$
We now use the following lemma.
\begin{lemma}
Let $\bg$ be a Gaussian random vector with each of the components being an independent standard Gaussian random variable. Let $\hat{v} = M\bg/\opnorm{M\bg}$. For any unit vector $v$, with probability $\ge 4/5$, 
\begin{align*}
|\la \hat{v}, v\ra |^2 \ge \frac{1}{1 + C\frac{\frnorm{M}^2 - \opnorm{\T{M}v}^2}{\opnorm{\T{M}v}^2}}
\end{align*}
for a large enough universal constant $C$.
\end{lemma}
\begin{proof}
    Since $v$ is a unit vector, we can write $\opnorm{M \bg}^2 = |\T{v}M\bg|^2 + \opnorm{(I-v\T{v})M\bg}^2$. Hence, we have
    \begin{align*}
        |\la \hat{v}, v\ra|^2 = \frac{|\T{v}Mg|^2}{\opnorm{M\bg}^2} = \frac{1}{1 + \frac{\opnorm{(I-v\T{v})M\bg}^2}{|\T{v}M\bg|^2}}.
    \end{align*}
    We now note that $\T{v}M\bg \sim N(0, \opnorm{\T{M}v}^2)$ and $\E[\opnorm{(I-v\T{v})M\bg}^2] = \text{tr}(\T{M}(I-v\T{v})M) = \frnorm{M}^2 - \opnorm{\T{M}v}^2$. By a union bound, with probability $\ge 4/5$, we have
    \begin{align*}
        \frac{\opnorm{(I-v\T{v})M\bg}^2}{|\T{v}M\bg|^2} \le C\frac{\frnorm{M}^2 - \opnorm{\T{M}v}^2}{\opnorm{\T{M}v}^2}
    \end{align*}
    for a large enough constant $C$. Therefore, with probability $\ge 4/5$, we get that \[|\la\hat{v}, v\ra|^2 \ge 
    \frac{1}{1 + C\frac{\frnorm{M}^2 - \opnorm{\T{M}v}^2}{\opnorm{\T{M}v}^2}}.\qedhere\]
\end{proof}
Applying the above lemma for $M = (\T{B_t}B_t) \cdots (\T{B_1}B_1)$ and $v = v_1$, we obtain
\begin{align*}
	|\la \hat{v}, v_1\ra|^2 \ge \frac{1}{1 + C' t\sqrt{\varepsilon}}
\end{align*}
with probability $\ge 4/5$.
\end{proof}

If $t = \Theta(\log d)$ and $1/\poly(d) \le \varepsilon \le c/(\log d)^2$, then the above lemma shows that $\hat{v}$ has a large correlation with the top singular vector $v_1$. Using this lemma, we show that Algorithm~\ref{alg:bounded-norms} can be used to obtain an approximation for $v_1$ in random order streams with bounded norms.
\begin{theorem}
	Let $\alpha \ge 1/\poly(d)$ be an accuracy parameter. Let $\eta$ be a parameter such that $\frac{\max_i\opnorm{a_i}^2}{\min_i \opnorm{a_i}^2} \le \eta$.
If the number of rows in the stream $n = \Omega(\alpha^{-4} \cdot \rho(A) \cdot \eta\cdot \log ^ 6 d)$, where $\rho(A) = \frnorm{A}^2/\opnorm{A}^2$ and the rows in the stream are ordered uniformly at random, then we can compute a vector $\hat{v}$ using the block power method that satisfies
	\begin{align*}
		|\la v_1, \hat{v}\ra|^2 \ge 1 - 3\alpha
	\end{align*}
	with probability $\ge 4/5$ if $\sigma_1(A)/\sigma_2(A) \ge 2$. The algorithm uses $O(d \cdot \polylog(d)/\alpha^4)$ bits of space.
 \label{thm:bounded-norms}
\end{theorem}
\begin{proof}
Set $\varepsilon = \alpha^2/C\log^2 d$ for a large enough constant $C$. Assuming $n = \Omega(\alpha^{-4} \rho \eta \log^6 d)$, we have $n = \Omega(\varepsilon^{-2} \rho \eta \log^2 d)$. Now consider the execution of Algorithm~\ref{alg:bounded-norms} on matrix $A$, with parameters $\eta$ and $\varepsilon$. Let $\rho = 2^{j}$ be such that $\rho(A)/2 \le \rho \le \rho(A)$, and consider the execution in the algorithm with parameter $\rho$. Using Theorem~\ref{thm:row-norm-sampling}, with probability $\ge 1 - 1/\poly(d)$, the algorithm computes $t$ matrices $\bB_1, \ldots, \bB_t$ such that for all $j \in [t]$,
\begin{align*}
    \opnorm{\frac{1}{p}\T{\bB_j}\bB_j - \T{A}A} \le \varepsilon\opnorm{A}^2.
\end{align*}
Noting that $\bz_\rho = (\T{\bB_t}\bB_t) \cdots (\T{\bB_1}\bB_1)\bg/\opnorm{(\T{\bB_t}\bB_t) \cdots (\T{\bB_1}\bB_1)\bg}$, by Lemma~\ref{lma:appx-quads-to-correlation}, we have with probability $\ge 9/10$ that
\begin{align*}
    \la \bz_{\rho}, v_1\ra^2 \ge \frac{1}{1+C't\sqrt{\varepsilon}} \ge 1 - \alpha.
\end{align*}
Thus, for $\rho$ which satisfies $\rho(A)/2 \le \rho \le \rho(A)$, the algorithm computes a vector $\bz_{\rho}$ that has a large correlation with the vector $v_1$. Since the algorithm does not know the exact value of $\rho$, it computes an approximation for $\opnorm{A\bz}^2$ for all $\bz \in \set{\bz_1, \bz_2, \bz_4, \ldots, \bz_d}$. First, we condition on the fact that with probability $\ge 1 - 1/\poly(d)$, for all $\bz_{i}$, $\opnorm{\bG A\bz_i}^2 = (1 \pm \varepsilon)\opnorm{A\bz_i}^2$. Since $\la \bz_{\rho}, v_1\ra^2 \ge (1 - \alpha)$, we note that $\opnorm{\bG A\bz_{\rho}}^2 \ge (1 - \varepsilon)(1 - \alpha)\sigma_1(A)^2$. Now, for the vector $\bz$ returned by the algorithm, we have
$\opnorm{\bA \bz}^2 \ge (1 - O(\varepsilon))(1 - \alpha)\sigma_1(A)^2$ which implies that 
\begin{align*}
    \la \bz, v_1\ra^2 \cdot \sigma_1(A)^2 + (1 - \la \bz, v_1\ra^2) \frac{\sigma_1(A)^2}{R} \ge \opnorm{A\bz}^2 \ge (1 - \alpha - O(\varepsilon))\sigma_1(A)^2
\end{align*}
and therefore $\la \bz, v_1\ra^2 \ge 1 - 3\alpha$ since $R \ge 2$.
\end{proof}


\subsection{Random Order Streams without Norm Bounds}
Assuming that the random order streams are long enough, Theorem~\ref{thm:bounded-norms} shows that if all the squared row norms are within an $\eta$ factor, then the block power method outputs a vector with a large correlation with the top eigenvector of the matrix $\T{A}A$. For general streams, the factor $\eta$ could be quite large and hence the algorithm requires very long streams to output an approximation to $v_1$. 

If there are no \emph{heavy} rows, i.e., rows with a Euclidean norm larger than $\frnorm{A}/\sqrt{d \cdot \polylog(d)}$, then the row norm sampling procedure in Theorem~\ref{thm:row-norm-sampling} can be used to convert any randomly ordered stream of rows into a uniformly random stream of rows that all have the same norm. The row norm sampling procedure computes a probability $p_i = \min(1, C\varepsilon^{-2}\opnorm{a_i}^2\log d/\opnorm{A}^2)$ and samples the row $a_i$ with probability $p_i$. If sampled, then the row $a_i$ is scaled by $1/\sqrt{p_i}$. From Theorem~\ref{thm:row-norm-sampling}, we have that the top eigenvector of the \emph{quadratic form} of the sampled-and-rescaled submatrix is a good approximation to the top eigenvector $\T{A}A$ when the gap $R$ is large enough. Suppose $p_i < 1$. If the row $a_i$ is sampled, we then have
\begin{align*}
    \opnorm{a_i/\sqrt{p_i}} = \frac{\varepsilon\opnorm{A}}{\sqrt{C\log d}}.
\end{align*}
Thus, if $p_i < 1$ for all $i$, then all the sampled-and-rescaled rows have the same Euclidean norm and therefore, we can run the algorithm from Theorem~\ref{thm:bounded-norms} by setting $\eta = 1$. Note that $p_i = 1$ only if $\opnorm{a_i}^2 \ge \varepsilon^2\opnorm{A}^2/C\log(d)$. Since we assumed that there are no heavy rows, there is no row with $p_i = 1$ as long as $\varepsilon \ge 1/\polylog(d)$. Thus, using Theorem~\ref{thm:bounded-norms} on the row norm sampled substream directly gives us a good approximation to the top eigenvector. However, in general, the streams can have rows with large Euclidean norm. We will now state our theorem and describe how such streams can be handled.
\begin{theorem}
    Let $A$ be an $n \times d$ matrix with its non-zero entries satisfying $1/\poly(d) \le |A_{i, j}| \le \poly(d)$, and hence representable using $O(\log d)$ bits of precision. Let $R = \sigma_1(A)^2/\sigma_2(A)^2$. Assume $2 \le R \le C_1\log^2 d$. Let $h$ be the number of rows in $A$ with norm at most $\frnorm{A}/\sqrt{d \cdot \polylog(d)}$, where $\polylog(d) = \log^{C_2} d$ for a large enough universal constant $C_2$. Given the rows of the matrix $A$ in a uniformly random order, there is an algorithm using $O((h + 1) \cdot d \cdot \polylog(d) \cdot \log n)$ bits of space and which outputs a vector $\hat{v}$ such that with probability $\ge 4/5$, $\hat{v}$ satisfies $\la \hat{v}, v_1\ra^2 \ge 1 - 8/\sqrt{R}$, where $v_1$ is the top eigenvector of the matrix $\T{A}A$.
    \label{thm:main-theorem-formal}
\end{theorem}
The key idea in proving this theorem is to partition the matrix $A$ into $A_{\text{heavy}}$ and $A_{\text{light}}$, where $A_{\text{heavy}}$ denotes the matrix with the heavy rows and $A_{\text{light}}$ denotes the matrix with the rest of the rows of $A$. Since we assume that there are at most $h$ heavy rows, we can store the matrix $A_{\text{heavy}}$ using $O(h \cdot d \cdot \polylog(d))$ bits of space. Now consider the following two cases: (i) $\opnorm{A_{\text{heavy}}} \ge (1 - \beta)\opnorm{A}$ or (ii) $\opnorm{A_{\text{heavy}}} < (1 - \beta)\opnorm{A}$ for some parameter $\beta$. In the first case, we can show that the top eigenvector $u$ of $\T{A_{\text{heavy}}}A_{\text{heavy}}$ is a good approximation for $v_1$. Since, we store the full matrix $A_{\text{heavy}}$, we can compute $u$ exactly at the end of the stream. Suppose $\opnorm{A_{\text{heavy}}} < (1 - \beta)\opnorm{A}$. By the triangle inequality, we have $\opnorm{A_{\text{light}}} > \beta\opnorm{A}$. If we set $\beta$ large enough compared to $1/R$, then we can show that the top eigenvector $u'$ of $\T{A_{\text{light}}}A_{\text{light}}$ is a good approximation of $v_1$. From the above discussion, since all the rows of $A_{\text{light}}$ are \emph{light}, we can obtain a stream using Theorem~\ref{thm:row-norm-sampling} such that all the rows have the same norm and additionally, the top eigenvector of this stream is a good approximation for $u'$ and therefore $v_1$. We then approximate the top eigenvector of the new stream using Theorem~\ref{thm:bounded-norms}. Setting $\beta$ appropriately, we show that this procedure can be used to compute a vector $\hat{v}$ satisfying $\la \hat{v}, v_1\ra^2 \ge 1 - O(1/\sqrt{R})$ proving the theorem. The formal proof is given below.
\begin{proof}
 Partition the matrix $A$ into $A_{\text{light}}$ and $A_{\text{heavy}}$, where $A_{\text{heavy}}$ is the submatrix with rows $a_i$ such that $\opnorm{a_i} > \frnorm{A}/\sqrt{d \cdot \polylog(d)}$ and $A_{\text{light}}$ is the remaining rows. From our assumption, the number of rows in $A_{\text{heavy}}$ is at most $h$.
Note that given a uniformly random stream of rows of $A$, we can obtain a uniformly random stream of rows of  $A_{\text{light}}$ by just filtering out the rows in $A_{\text{heavy}}$.

Suppose, $\opnorm{A_{\text{heavy}} \cdot v_1} \ge (1 - \beta)\opnorm{A}$ for a parameter $\beta$ to be chosen later. Let $v_1'$ be the top singular vector of the matrix $A_{\text{heavy}}$. Note
\begin{align*}
	\opnorm{A \cdot v_1'}^2 \ge \opnorm{A_{\text{heavy}} \cdot v_1'}^2 \ge \opnorm{A_{\text{heavy}} \cdot v_1}^2 \ge (1 - \beta)^2\opnorm{A}^2
\end{align*}
and therefore we have $\la v_1', v_1\ra^2 \ge 1 - 4\beta$, assuming $R \ge 2$. Thus, while processing the stream, we can store all the heavy rows and at the end of the stream compute the top right singular vector of $A_{\text{heavy}}$, in order to obtain a good approximation for $v_1$.

Suppose $\opnorm{A_{\text{heavy}} \cdot v_1} \le (1 - \beta)\opnorm{A}$. This implies $\opnorm{A_{\text{light}} \cdot v_1}^2 \ge \opnorm{A}^2 - \opnorm{A_{\text{heavy}} \cdot v_1}^2 \ge \beta \cdot \opnorm{A}^2$. If we set $\beta \ge 2/R$, we have
\begin{align*}
	\frac{\sigma_1(A_{\text{light}})^2}{\sigma_2(A_{\text{light}})^2} \ge \frac{\beta\opnorm{A}^2}{\sigma_2(A)^2} \ge 2.
\end{align*}
Let $v_1'$ be the top singular vector of $A_{\text{light}}$. We will describe how to approximate $v_1'$. Consider applying the row norm sampling procedure with parameter $\varepsilon$ to the matrix $A_{\text{light}}$. Given a row $a_i \in {A_{\text{light}}}$ the corresponding sampling probability $p_i$ is given by
\begin{align*}
    p_i = \frac{C\log d \cdot \opnorm{a_i}^2}{\varepsilon^2\opnorm{A_{\text{light}}}^2} \le \frac{C\log d \cdot \frnorm{A}^2/(d \cdot \polylog(d))}{\varepsilon^2\beta^2\opnorm{A}^2} \le \frac{C}{\varepsilon^2\beta^2\polylog(d)}.
\end{align*}
Assuming that $\varepsilon^2 \beta^2 \ge 1/\polylog(d)$, we obtain that $p_i < 1$ for all the rows in the matrix $A_{\text{light}}$. Let $\bB_{\text{light}}$ be the matrix obtained after applying the row norm sampling procedure to the matrix $A_{\text{light}}$. Note that $\rho(\bB_{\text{light}}) \approx \rho(\bA_{\text{light}})$ and the number of rows in $\bB_{\text{light}}$ is $\Theta(\rho(A_{\text{light}}) \cdot \log d \cdot \varepsilon^{-2})$, and therefore $\Theta(\rho(\bB_{\text{light}}) \cdot \log d \cdot \varepsilon^{-2})$. Setting $\varepsilon = \alpha^2 / \log^{5/2} d$, we obtain that the number of rows in the matrix $\bB_{\text{light}}$ is $\Theta(\alpha^{-4} \cdot \rho(\bB_{\text{light}}) \cdot \log^6 d)$ and thus assuming $\varepsilon^2\beta^2 = \alpha^4\beta^2/\log^5 d \ge 1/\polylog(d)$, we can use Theorem~\ref{thm:bounded-norms} to obtain a vector $\hat{v}$ satisfying 
\begin{align*}
    \la \hat{v}, v_1'\ra^2 \ge 1 - 3\alpha.
\end{align*}
We will now show that $v_1'$ has a large correlation with $v_1$ which then implies $\hat{v}$ has a large correlation with $v_1$. Since $\opnorm{A_{\text{light}}} \ge \opnorm{A} - \opnorm{A_{\text{heavy}}} \ge \beta\opnorm{A}$,
$
\opnorm{A_{\text{light}}}^2 = \opnorm{A_{\text{light}} \cdot v_1'}^2 \ge \beta \opnorm{A}^2.
$
Consider the following upper bound on $\opnorm{A_{\text{light}} \cdot v_1'}^2$:
\begin{align*}
	\opnorm{A_{\text{light}}}^2 = \opnorm{A_{\text{light}} \cdot v_1'}^2 &= \opnorm{A_{\text{light}} \cdot (\la v_1', v_1\ra \cdot v_1 + (I - v_1\T{v_1})v_1')}^2\\
	& = \opnorm{\la v_1, v_1'\ra A_{\text{light}} \cdot v_1 + A_{\text{light}} (I-v_1\T{v_1})v_1'}^2\\
	&\le (1 + \theta) \cdot \la v_1, v_1'\ra^2 \cdot \opnorm{A_{\text{light}} \cdot v_1}^2 + (1 + 1/\theta) \cdot \opnorm{A_{\text{light}} (I - v_1\T{v_1})v_1'}^2
\end{align*}
for any $\theta > 0$. Using the fact that the rows of the matrix $A_{\text{light}}$ are a subset of the rows of the matrix $A$ and that $\opnorm{A(I-v_1\T{v_1})} = \sigma_2(A) = \sigma_1(A)/\sqrt{R}$, we have
\begin{align*}
	\opnorm{A_{\text{light}}}^2 &\le (1 + \theta) \cdot \la v_1, v_1'\ra^2 \cdot \opnorm{A_{\text{light}}}^2 + (1 + 1/\theta) \cdot \frac{\sigma_1^2}{R} \cdot (1 - \la v_1, v_1'\ra^2)\\
	&= \la v_1, v_1'\ra^2 ((1+\theta) \cdot \opnorm{A_{\text{light}}}^2 - (1 + 1/\theta)\sigma_1^2/R) + (1 + 1/\theta) \cdot \sigma_1^2/R
\end{align*}
which implies
\begin{align*}
	\la v_1, v_1'\ra^2 \ge \frac{\opnorm{A_{\text{light}}}^2 - (1+1/\theta)\cdot \sigma_1^2/R}{(1+\theta)\opnorm{A_{\text{light}}}^2 - (1 + 1/\theta)\sigma_1^2/R} &= 1 - \frac{\theta \cdot \opnorm{A_{\text{light}}}^2}{(1+\theta)\opnorm{A_{\text{light}}}^2 - (1 + 1/\theta)\sigma_1^2/R}\\
	&\ge 1 - \frac{\theta}{1 + \theta - (1 + 1/\theta)/R\beta}
\end{align*}
using the fact that $\opnorm{A_{\text{light}}}^2 \ge \beta^2 \sigma_1^2$. Now assuming $R\beta \ge 1$ and picking $\theta = 2/(R\beta - 1)$, we obtain
\begin{align*}
	\la v_1, v_1'\ra^2 \ge 1 - \frac{4R\beta}{(1+R\beta)^2} \ge 1 - \frac{4}{R\beta}.
\end{align*}
We therefore have
\begin{align}
    \la \hat{v}, v_1\ra^2 \ge 1 - \frac{4}{R\beta} - 4\alpha.
\end{align}
Setting $\beta = 1/\sqrt{R}$ and $\alpha = 1/\sqrt{R}$, we satisfy all the requirements assuming that $R \le \polylog(d)$ and obtain a vector $\hat{v}$ satisfying $\la \hat{v}, v_1\ra^2 \ge 1 - 8/\sqrt{R}$. When $\opnorm{A_{\text{heavy}}} \ge (1 - \beta)\opnorm{A}$, we already have a vector $v' = \text{top eigenvector of $A_{\text{heavy}}$}$ that satisfies $\la \hat{v}, v_1\ra^2 \ge 1 - 4\beta \ge 1 - 4/\sqrt{R}$. Thus, in both the cases, we obtain a vector $\hat{v}$ satisfying $\la \hat{v}, v_1\ra^2 \ge 1 - O(1/\sqrt{R})$.

The procedure described requires knowing the approximate values of $\frnorm{A}$, $\opnorm{A_{\text{light}}}$. Since, we assume that all the non-zero entries of the matrix have an absolute value at least $1/\poly(d)$ and at most $\poly(d)$, the values $\frnorm{A}, \opnorm{A_{\text{light}}}$ lie in the interval $[1/\poly(d), \poly(nd)]$. Hence, using $O(\log nd)$ guesses each for $\frnorm{A}$ and $\opnorm{A_{\text{light}}}$ and using a Gaussian sketch of $A$ similar to that in Algorithm~\ref{alg:bounded-norms}, we can obtain a vector satisfying the guarantees in the theorem.
\end{proof}

\section{Lower Bounds}
Our algorithm uses $\tilde{O}(h \cdot d)$ space when the number of heavy rows in the stream is $h$. We want to argue that it is nearly tight. We show the following theorem.
\begin{theorem}
    Given a dimension $d$, let $h$ and $R$ be arbitrary with $R \le h \le d$ and $R^2 \cdot h = O(d)$. Consider an algorithm $\calA$ with the following property:
    \begin{center}
    Given any fixed matrix $n \times d$ matrix $A$ with $O(h)$ heavy rows and gap $\sigma_1(A)^2/\sigma_2(A)^2 \ge R$, in the form of a uniform random order stream, the algorithm $\calA$ outputs a unit vector $\hat{v}$ such that, with probability $\ge 1 - (1/2)^{4R+4}$ over the randomness of the stream and the internal randomness of the algorithm,   
        $|\la \hat{v}, v_1\ra|^2 \ge 1 - c/R^2.$
    \end{center}
If $c$ is a small enough constant, then the algorithm $\calA$ must use $\Omega(h \cdot d/R)$ bits of space.
\label{thm:lowerbound-theorem-formal}
\end{theorem}

The theorem shows that a streaming algorithm must use $\Omega(hd/R)$ bits of space assuming that with high probability, it outputs a vector with a large enough correlation with the top eigenvector of $\T{A}A$ when the rows are given in a random order stream. Our proof uses the same lower bound instance as that of \cite{price2023spectral}. The key difference from their proof is that our lower bound must hold against random order streams.
\begin{proof}
For each $i \in [h]$, let $\bx_1, \ldots, \bx_h$ be drawn independently and uniformly at random from $\set{+1, -1}^d$. Let $\bi \sim [h]$ be drawn uniformly at random, and for an integer $k$ to be chosen later, let $\by_1, \ldots, \by_k \in \R^d$ be vectors that share the first $(1-\gamma)d$ coordinates with the vector $\bx_{\bi}$. Each of the last $\gamma \cdot d$ elements of each of $\by_1, \ldots, \by_k$ are sampled uniformly at random from the set $\set{+1, -1}$. Define $\bz_1, \ldots, \bz_{h+k}$ such that for $j \le h$, $\bz_j = \bx_j$ and for $j > h$, let $\bz_j = \by_{j - h}$.

Now consider the stream $\bz_1, \ldots, \bz_{h+k}$. Price and Xun argue that when $k \ge 4R$, the gap of this stream is at least $R$ with large probability over the randomness used in the construction of the stream. Let $\bpi : [h+k] \rightarrow [h+k]$ be a uniformly random permutation independent of $\bi$. Consider the following event $\calE$:
\begin{center}
	$\bpi(\bi) \le h/2$ and $\bpi(h+1), \ldots, \bpi(h+k) > h/2$.
\end{center} 
We have that the probability of the event $\calE$ is
\begin{align*}
	\frac{h/2 + k}{h + k} \cdot \frac{h/2 + k - 1}{h + k - 1} \cdots \frac{h/2 + 1}{h + 1} \cdot \frac{h/2}{h} \ge (1/2)^{k+1}.
\end{align*} 
Let $S_{\bi}$ be the set of permutations $\pi$ that satisfy the above event. Therefore we have $\Pr_{\bpi}[\bpi \in S_{\bi}] \ge (1/2)^{k+1}$. If the probability of failure, $\delta$, of the algorithm $\calA$ satisfies $\delta \le (1/2)^{k+4}$, we have that
\begin{align*}
	\Pr_{\bpi,\ \text{internal randomness}}[\text{$\calA$ succeeds on $\bz_{\bpi(1)}, \ldots, \bz_{\bpi(h+k)}$} \mid \bpi \in S_{\bi}]\ \ge \frac{3}{4}.
\end{align*}
Let $\bs_{\text{mid}}$ be the state of the algorithm after $h/2$ steps and $\bs_{\text{fin}}$ be the final state of the algorithm. The randomness in $\bs_{\text{fin}}$ is from the following sources: (i) randomness of the vectors $\bx_1, \ldots, \bx_h$, (ii) the index $\bi \in [h]$, (iii) the vectors $\by_1, \ldots, \by_k$, (iv) the permutation $\bpi$, and (v) the internal randomness of the algorithm. From here on, condition on the event $\calE$, i.e., that the permutation $\bpi \in S_{\bi}$. We will not explicitly mention that all entropy and information terms in the proof are conditioned on $\calE$. Since $\bpi(\bi) \le h/2$, we have
\begin{align*}
    \bs_{\text{fin}}\text{ is conditionally independent of }x_{\bi}[(1-\gamma)\cdot d + 1: d] \text{ given $\bs_{\text{mid}}$}.
\end{align*}
Using the data processing inequality, we obtain that
\begin{align*}
	I(\bs_{\text{mid}}; \bx_{\bi}[(1-\gamma)\cdot d + 1:d]) \ge I(\bs_{\text{fin}}; \bx_{\bi}[(1-\gamma)\cdot d + 1: d]).
\end{align*}
When $h \le cd/R^2$, $k = 4R$, $\gamma = 1/4$  and $\varepsilon \le c/k^2$ for a small constant, we have as in the proof of Theorem~1.5 in \cite{price2023spectral} that,
\begin{align*}
	I(\bs_{\text{fin}}; \bx_{\bi}[(1-\gamma) \cdot d+1 : d]) \ge \Omega(d/R)
\end{align*}
which now implies
\begin{align*}
	I(\bs_{\text{mid}}; \bx_{\bi}[(1-\gamma) \cdot d+1 : d]) \ge \Omega(d/R).
\end{align*}
Note that conditioned on the event $\calE$, the distribution of $\bi$ is uniform over $\set{\bpi^{-1}(1), \ldots, \bpi^{-1}(h/2)}$. We now prove the following lemma:
\begin{lemma}
    Let $\bY_1, \ldots, \bY_{\ell}$ be independent random variables. Let $\bi \sim [\ell]$ be a uniform random variable independent of $\bX$. We have
    \begin{align*}
        I(\bX\, ;\, \bY_1) + \cdots + I(\bX\, ;\, \bY_{\ell}) \ge \ell \cdot  (I(\bX; \bY_{\bi}) - \log_2 \ell).
    \end{align*}
\end{lemma}
\begin{proof}
By definition, we have
\begin{align*}
    I(\bX\, ; \bY_{\bi}) = H(\bY_{\bi}) - H(\bY_{\bi} \mid \bX).
\end{align*}
Now, we note that $H(\bY_{\bi}) \le H(\bY_{\bi}, \bi) = H(\bi) + H(\bY_{\bi} \mid \bi) = \log_2 \ell + \frac{H(\bY_1) + \cdots + H(\bY_{\ell})}{\ell}$. We now lower bound $H(\bY_{\bi} \mid \bX)$. Since conditioning always decreases entropy, we obtain
\begin{align*}
    H(\bY_{\bi} \mid \bX) \ge H(\bY_{\bi} \mid \bi, \bX).
\end{align*}
As $\bX$ is independent of $\bi$, we have
\begin{align*}
    H(\bY_{\bi} \mid \bX) \ge H(\bY_{\bi} \mid \bi, \bX) = \frac{H(\bY_1 \mid \bX) + \cdots + H(\bY_{\ell} \mid \bX)}{\ell}
\end{align*}
which then implies
\begin{align*}
    I(\bX\, ;\, \bY_{\bi}) &\le H(\bi) + \frac{H(\bY_1) + \cdots + H(\bY_{\ell})}{\ell} - \frac{H(\bY_1 \mid \bX) + \cdots + H(\bY_{\ell} \mid \bX)}{\ell}\\
    &\le H(\bi) + \frac{I(\bX\,;\,\bY_1) + \cdots + I(\bX\,;\, \bY_{\ell})}{\ell}.
\end{align*}
Since $H(\bi) = \log_2 \ell$, we have the proof.
\end{proof}
Using this lemma, 
\begin{align*}
 &I(\bs_{\text{mid}}; \bx_{\bpi^{-1}(1)}[(1-\gamma) \cdot d + 1:d]) + \cdots + I(\bs_{\text{mid}}; \bx_{\bpi^{-1}(h/2)}[(1-\gamma) \cdot d + 1:d])\\	
 &= (h/2) \cdot I(\bs_{\text{mid}}; \bx_{\bi}[(1-\gamma) \cdot d + 1 : d] - \log_2 (h/2))\\
 &\ge \Omega(hd/R) - h\log_2 h.
\end{align*}
\begin{lemma}
	If $\bX, \bY$ are independent, then $I(\bZ\, ;\, (\bX, \bY)) \ge I(\bZ\, ;\, \bX) + I(\bZ\, ;\, \bY)$.
\end{lemma}
\begin{proof}
	\begin{align*}
		I(\bZ\, ;\, (\bX, \bY)) &= H((\bX, \bY)) - H((\bX, \bY) \mid \bZ)\\
		&= H(\bX) + H(\bY) - H((\bX, \bY) \mid \bZ).
	\end{align*}
	Now, we note that for any three random variables $\bX, \bY, \bZ$, we have $H((\bX, \bY) \mid \bZ) \le H(\bX \mid \bZ) + H(\bY \mid \bZ)$ which proves the lemma.
\end{proof}
Using the independence of $\bx_1, \ldots, \bx_h$ conditioned on the event $\calE$, we obtain
\begin{align*}
	I(\bs_{\text{mid}};(\bx_{\bpi^{-1}(1)}[(1-\gamma) \cdot d + 1:d],\ldots,\bx_{\bpi^{-1}(h/2)}[(1-\gamma) \cdot d + 1:d])) \ge \Omega(hd/R) - h \log_2 h
\end{align*}
which then implies
\begin{align*}
	H(\bs_{\text{mid}}) \ge \Omega(hd/R)
\end{align*}
using the fact that $R^2 \cdot h = O(d)$.
Finally, we have $\max |\bs_{\text{mid}}| \ge \Omega(hd/R)$. Here $|\bs_{\text{mid}}|$ is the number of bits used in the representation of the state $\bs_{\text{mid}}$.
\end{proof}
\section{Improving the Gap Requirements in the Algorithm of Price and Xun}
\subsection{Arbitrary Order Streams}
As discussed in Section~\ref{sec:sampling-for-row-reduction}, we can guess an approximation of $\opnorm{A}^2$ in powers of $2$ and sample at most $O(d\log d/\varepsilon^2)$ rows in the stream to obtain a matrix $\bB$, in the form of a stream, satisfying
$
    \opnorm{\T{\bB}\bB - \T{A}A} \le \varepsilon\opnorm{A}^2,
$
with a large probability.
Using Weyl's inequalities, we obtain that
\begin{align*}
    \sigma_2(\T{\bB}\bB) \le \sigma_2(\T{A}A) + \varepsilon\opnorm{A}^2
\quad
\text{and}
\quad
    \sigma_1(\T{\bB}\bB) \ge (1-\varepsilon)\sigma_1(\T{A}A)
\end{align*}
implying $R' = \sigma_1(\bB)^2/\sigma_2(\bB)^2 \ge (1-\varepsilon)/(1/R + \varepsilon)$. For $\varepsilon = 1/(2R) \le 1/2$, we note $R' \ge R/3$. Let $n' = O(R^2 \cdot d\log d)$ be the number of rows in the matrix $\bB$ and note that $R' = \Omega(\log n' \cdot \log d)$ assuming $R = \Omega(\log^2 d)$. Hence, running the algorithm of Price and Xun on the rows of the matrix $\bB$, we compute a vector $\hat{v}$ for which
\begin{align*}
    |\la \hat{v}, v_1'\ra|^2 \ge 1 - \frac{\log d}{CR'} - \frac{1}{\poly(d)}
\end{align*}
with a large probability, where $v_1'$ is the top eigenvector of the matrix $\T{\bB}\bB$. We now note that if $v_1$ denotes the top eigenvector of the matrix $\T{A}A$, then $|\la v_1, v_1'\ra|^2 \ge  1 - O(1/R)$ which therefore implies that with a large probability,
\begin{align*}
    |\la \hat{v}, v_1\ra|^2 \ge 1 - \frac{\log d}{CR}.
\end{align*}
Thus, sub-sampling the stream using row norm sampling and then running the algorithm of \cite{price2023spectral}, we obtain an algorithm for arbitrary order streams with a gap $R = \Omega(\log^2 d)$.
\subsection{Random Order Streams}
Lemma 3.5 in \cite{price2023spectral} can be tightened when the rows of the stream are uniformly randomly ordered. Specifically, we want to bound the following quantity:
\begin{align*}
	\sum_{i=1}^n \la a_i, P\hat{v}_{i-1}\ra^2
\end{align*}
where $P = I - v_1\T{v_1}$ denotes the projection away from the top eigenvector, and $\hat{v}_{i-1}$ is a function of $v_1, a_1, \ldots, a_{i-1}$. We have
\begin{align*}
	\E[\la a_i, P\hat{v}_{i-1}\ra^2] = \E[\E[\la a_i, P\hat{v}_{i-1}\ra^2 \mid a_1, \ldots, a_{i-1}]].
\end{align*}
Given that the first $i-1$ rows are $a_1, \ldots, a_{i-1}$, assuming uniform random order, we have 
\begin{align*}
\E[\la a_i, P\hat{v}_{i-1}\ra^2 \mid a_1, \ldots, a_{i-1}] &= \frac{1}{n-i + 1}\T{\hat{v}_{i-1}}P(\T{A}A - a_1\T{a_1} - \cdots - a_{i-1}\T{a_{i-1}})P\hat{v}_{i-1}\\
&\le \frac{\sigma_2(A)^2}{n-i+1}.
\end{align*}
Hence $\E[\la a_i, P\hat{v}_{i-1}\ra^2] \le \sigma_2(A)^2/(n-i+1)$ and
$
	\E[\sum_{i=1}^n \la a_i, P\hat{v}_{i-1}\ra^2] \le \sigma_2(A)^2(1 + \log n).
$
Price and Xun define $\eta \cdot \sigma_2(A)^2$ as $\sigma_2$ and in that notation, we obtain $\eta\sum_{i=1}^n \la a_i, P\hat{v}_{i-1}\ra^2 \le 10\sigma_2(1+\log n)$ with probability $\ge 9/10$ by Markov's inequality. In the proof of Lemma~3.6 in \cite{price2023spectral}, if $\sigma_1/\sigma_2 \ge 20(1 + \log_2 n)$, we obtain $\log \opnorm{v_n} \gtrsim \sigma_1$. Now, $\sigma_1 \ge O(\log d)$ ensures that the Proof of Theorem~1.1 in their work goes through.

Using the row-norm sampling analysis from the previous section, we can assume $n = \poly(d)$ and therefore a gap of $O(\log d)$ between the top two eigenvalues of $\T{A}A$ is enough for Oja's algorithm to output a vector with a large correlation with the top eigenvector in random order streams.
\section{Hard Instance for Oja's Algorithm}
At a high level, the algorithm of \cite{price2023spectral} runs Oja's algorithm with different learning rates $\eta$ and in the event that the norm of the output vector with each of the learning rates $\eta$ is small, then the row with the largest norm is output. The algorithm is simple and can be implemented using an overall space of $O(d \cdot \polylog(d))$ bits. 

The algorithm initializes $z_0 = \bg$ where $\bg$ is a random Gaussian vector. The algorithm streams through the rows $a_1, \ldots, a_n$ and performs the following operation
\begin{align*}
    z_i \gets z_{i-1} + \eta \cdot \la z_{i-1}, a_i\ra a_i.
\end{align*}
The algorithm computes the smallest learning rate $\eta$ when $\opnorm{z_n}$ is large enough, and then outputs either $z_n/\opnorm{z_n}$ or $\bar{a}/\opnorm{\bar{a}}$ as an approximation to the eigenvector of the matrix $\T{A}A$. Here $\bar{a}$ denotes the row in $A$ with the largest Euclidean norm.

The following theorem shows that at gaps $\le O(\log d/\log\log d)$, we cannot use Oja's algorithm with a fixed learning rate $\eta$ to obtain constant correlation with the top eigenvector.

\begin{theorem}
    Given dimension $d$, a constant $c > 0$, a parameter $M$, for all gap parameters $R = O_c(\log d/\log\log d)$ there is a stream of vectors $a_1, \ldots, a_n \in \R^d$ with $n = O(R + M)$ such that:
    \begin{enumerate}
        \item $\sigma_1(A)^2/\sigma_2(A)^2 \ge R/2$, and
        \item Oja's algorithm with any learning rate $\eta < M$ fails to output a unit vector $\hat{v}$ that satisfies, with probability $\ge 9/10$,
        \begin{align*}
            |\la \hat{v}, v_1\ra| \ge c
        \end{align*}
        where $v_1$ is the top eigenvector of the matrix $\T{A}A$.
    \end{enumerate}
    Moreover, the result holds irrespective of the order in which the vectors $a_1, \ldots, a_n$ are presented to the Oja's algorithm. We will additionally show that even keeping track of the largest norm vector is insufficient to output a vector that has a large correlation with $v_1$.
    \label{thm:ojas-lowerbound-formal}
\end{theorem}
\begin{proof}
Our instance consists of the following vectors: 
\begin{enumerate}
    \item $R$ copies of the vector $(1/\sqrt{R})e_1$,
    \item 1 copy of the vector $(1/\sqrt{R-\varepsilon})e_2$, and
    \item $\alpha$ copies of the vector $(1/\sqrt{\alpha \cdot R})e_3$.
\end{enumerate}
where $\alpha = 2M$. Let $A$ be a matrix with rows given by the stream of vectors defined above. We note that the matrix $A$ has rank $3$ and the non-zero eigenvalues of the matrix $\T{A}A$ are $1, 1/(R-\varepsilon), 1/R$ and therefore the \emph{gap} $\lambda_1(\T{A}A)/\lambda_2(\T{A}A) = R - \varepsilon$. The top eigenvector of the matrix $\T{A}A$ is $e_1$ and the row with the largest norm is $(1/\sqrt{R-\varepsilon})e_2$. Thus, the row with the largest norm is not useful to obtain correlation with the true top eigenvector $e_1$.

Consider an execution of Oja's algorithm with a learning rate $\eta$ on the above stream of vectors. The final vector $z_n$ can be written as
\begin{align*}
    z_n = \left(I + \frac{\eta}{R}e_1\T{e_1}\right)^R\left(I + \frac{\eta}{R\alpha}e_3\T{e_3}\right)^{\alpha}\left(I + \frac{1}{R-\varepsilon}e_2\T{e_2}\right)v_0.
\end{align*}
For $j \in [d]$, let $z_{ij}$ denote the $j$-th coordinate of the vector $z_i$ so that we have
\begin{align*}
    z_{n1} &= \left(1 + \frac{\eta}{R}\right)^R \cdot z_{01},\\
    z_{n2} & = \left(1 + \frac{\eta}{R-\varepsilon}\right) \cdot z_{02},\quad \text{and}\\
    z_{n3} &= \left(1 + \frac{\eta}{R\alpha}\right)^\alpha \cdot z_{03}.
\end{align*}
We note that $z_{nj} = z_{0j}$ for all $j > 3$. Since $\alpha = 2M$, we have $\eta/R\alpha \le 1/2$ and therefore $(1 + \eta/R\alpha) \ge \exp(\eta/2R\alpha)$ and $(1+\eta/R\alpha)^{\alpha} \ge \exp(\eta/2R)$. 

Recall that we want to show that $|\la z_n, e_1\ra| < c\opnorm{z_n}$ with a large probability. Suppose otherwise and that with probability $\ge 1/10$, we have $|\la z_n, e_1\ra| > c\opnorm{z_n} > c\opnorm{(0,\ 0,\ 0,\ z_{04},\ \ldots,\ z_{0d})}$. 

Since, $z_0$ is initialized to be a random Gaussian, we have $\opnorm{(0, 0, 0, z_{04}, \ldots, z_{0d})} \ge \sqrt{d}/2$ with probability $1 - \exp(-d)$. Thus, we have with probability $\ge 1/11$ that,
\begin{align*}
    |z_{n1}| \ge c\sqrt{d}/2
\end{align*}
which implies the learning rate must satisfy
\begin{align*}
    (1 + \eta/R)^R \ge c'\sqrt{d}/2
\end{align*}
since $|z_{01}| \le 10$ with probability $\ge 99/100$. Hence $\eta \ge R((c'd^{1/2})^{1/R} - 1)$. Now consider $|\la z_n, e_3\ra|/|\la z_{n}, e_1 \ra|$. We have
\begin{align*}
    \frac{|\la z_n, e_3\ra|}{|\la z_n, e_1\ra|} = \frac{\exp(\eta/R)}{(1+\eta/R)^R} \cdot \frac{|z_{03}|}{|z_{01}|}.
\end{align*}
With probability $\ge 95/100$, we have $1/C \le |z_{03}|/|z_{01}| \le C$ for a large enough constant $C$. We now consider the expression
\begin{align*}
    \frac{\exp(\eta/R)}{(1+\eta/R)^R}.
\end{align*}
The expression is minimized at $\eta = R^2 - R$ and is increasing in the range $\eta \in [R^2 - R, \infty)$. When, $R = O(\log d/\log\log d)$, we have that $R^2 - R \le R((c'd^{1/2})^{1/R} - 1)$ and therefore for all $\eta \ge R((c'd^{1/2})^{1/R} - 1)$, we have
\begin{align*}
    \frac{\exp(\eta/R)}{(1+\eta/R)^R} \ge \frac{\exp((c'd^{1/2})^{1/R})}{e \cdot c'd^{1/2}}.
\end{align*}
When $R = O(\log d/\log\log d)$, we have
\begin{align*}
    \frac{\exp(\eta/R)}{(1+\eta/R)^R} \ge \poly(d)
\end{align*}
which then implies $|\la z_n, e_3\ra| \ge |\la z_n, e_1\ra| \cdot \poly(d)/C$ with probability $\ge 95/100$ which contradicts our assumption that $|\la z_n, e_1\ra| \ge c\opnorm{z_n}$.
\end{proof}
\section*{Acknowledgements}
The authors were supported in part by a Simons Investigator Award and NSF CCF-2335412. D. Woodruff was visiting Google Research while performing this work.

\bibliographystyle{plainnat}
\bibliography{main}

\end{document}